\documentclass{article}
\usepackage{arxiv}
\usepackage[dvipdfmx]{graphicx}
\usepackage[utf8]{inputenc} 
\usepackage[T1]{fontenc}    
\usepackage{hyperref}       
\usepackage{url}            
\usepackage{booktabs}       
\usepackage{amsfonts}       
\usepackage{microtype}      
\usepackage{graphicx}
\usepackage{enumerate}\usepackage{hyperref}
\usepackage{amsfonts}
\usepackage{bussproofs}
\usepackage{mathptmx}
\usepackage{latexsym}
\usepackage{amsmath}
\usepackage{amsthm}
\theoremstyle{remark}
\newtheorem{remark}{Remark}
\theoremstyle{definition}
\newtheorem{definition}{Definition}
\theoremstyle{theorem}
\newtheorem{theorem}{Theorem}
\theoremstyle{theorem}
\newtheorem{lemma}{Lemma}
\theoremstyle{remark}
\newtheorem{example}{Example}

\title{Quantum Modal Logic}

\newif\ifuniqueAffiliation
\uniqueAffiliationtrue

\ifuniqueAffiliation 
\author{ \href{https://orcid.org/0000-0001-8878-9824}{\hspace{1mm}Kenji ~Tokuo}\\
Department of Information Engineering, Oita College\\
National Institute of Technology\\
Oita 870-0152\\
Japan\\
\texttt{tokuo@oita-ct.ac.jp} \\
}
\else
\usepackage{authblk}

\author[1]{%
\href{https://orcid.org/0000-0000-0000-0000}{\usebox{\orcid}\hspace{1mm}David S.~Hippocampus\thanks{\texttt{hippo@cs.cranberry-lemon.edu}}}%
}
\author[1,2]{%
\href{https://orcid.org/0000-0000-0000-0000}{\usebox{\orcid}\hspace{1mm}Elias D.~Striatum\thanks{\texttt{stariate@ee.mount-sheikh.edu}}}%
}
\affil[1]{Department of Computer Science, Cranberry-Lemon University, Pittsburgh, PA 15213}
\affil[2]{Department of Electrical Engineering, Mount-Sheikh University, Santa Narimana, Levand}
\fi
\date{October 14, 2024}

\begin{document}
  \maketitle
  \footnotetext[2]{This is a pre-copyedited, author-produced version of an article accepted for publication in Logic Journal of the IGPL following peer review. The version of record
  Kenji Tokuo, Quantum modal logic, {\it Logic Journal of the IGPL}, 2024
  is available online at: https://doi.org/10.1093/jigpal/jzae123
  }

  \begin{abstract}
    A modal logic based on quantum logic is formalized in its simplest possible form.
    Specifically, a relational semantics and a sequent calculus are provided, and the soundness and the completeness theorems connecting both notions are demonstrated.
    This framework is intended to serve as a basis for formalizing various modal logics over quantum logic, such as quantum alethic logic, quantum temporal logic, quantum epistemic logic, and quantum dynamic logic.
  \end{abstract}

  \markboth{}
  {Quantum Modal Logic}
  \clearpage

  \section{Introduction}
  \label{sec:int}

  \subsection{Aims}
  \label{sub:aim}
  In this study, we formalize a modal logic based on quantum logic from a general and abstract standpoint.
  Specifically, we provide a relational semantics and a sequent calculus, and establish the soundness and the completeness theorems connecting both notions.
  The framework is intended to rigorously address various modal logics over quantum logic, such as quantum temporal logic, quantum epistemic logic, and quantum dynamic logic, thereby advancing research in these areas.
  \subsection{Background}
  \label{sub:bac}
  Traditional quantum logic, originated by Birkhoff and von Neumann \cite{Birk1936}, was formulated as orthomodular logic and orthologic in the 1960s and has been a part of non-classical logic research up to the present day.
  The relationship between quantum logic and modal logic was already discussed in the pioneering period of quantum logic research.
  Dishkant \cite{Dish1978} embedded Mackey's axioms \cite{Mack1963} into the infinite multi-valued logic of \L{}ukasiewicz by considering probabilities as truth values. The resulting system can be viewed as an extension of \L{}ukasiewicz's logic to modal logic.
  The modal symbol $Q$ used by him has the meaning of ``is confirmed''.
  Mittelstaedt \cite{Mitt1979} adopted quantum logic as the object logic and ordinary intuitionistic or classical logic as the meta-logic, then introduced modalities on the meta-logic side.
  In quantum theory, it is reasonable to introduce meta-statements to express that the truth of a proposition is ``certain'' or ``undeniable'' through subsequent time evolution, as the information possessed by the system is not perfect in its initial state.
  In addition to them, he introduced the modalities of ``contingent'' and ``objective'' into the meta-logic.
  Fraassen \cite{Fraa1981} stated that the propositions studied in quantum logic are modal: they give information first and foremost about what can and what must happen, and only indirectly about what actually does happen.
  Relations used to define physical necessity may include the following:
  \begin{itemize}
    \item $R_1(i,j)$: all physical laws of $i$ are true in $j$.
    \item $R_2(i,j)$: all physical laws of $i$ are also physical laws of $j$.
    \item $R_3(i,j)$: all and only the physical laws of $i$ are physical laws of $j$.
  \end{itemize}
  Setting $R_1$ as reflexive, $R_2$ reflexive and transitive, and $R_3$ reflexive, transitive, and symmetric, he considered necessity operators corresponding to these relations.
  Pessoa \cite{Pess2005} investigated an experiment in which a single photon passes through the Mach-Zehnder interferometer. He argued that non-distributive logic is unsatisfactory for interpreting this experiment and proposed a modal logic interpretation where $\Box\alpha$ indicates that ``a measurement at position $\Box\alpha$ would necessarily detect the presence of a photon,'' offering a solution to the problem.

  Diverse quantum modal logics have also been proposed for various purposes.
  Yu \cite{Yu2019} introduced quantum modal logic for quantum concurrent programs.
  This logic is semantically defined and includes temporal modality such as next, until, almost surely until, false, true, eventually, almost surely eventually, and always, while it does not include a negation operation.
  Baltag and Smets \cite{Balt2010} proposed a quantum epistemic logic.
  Traditionally, epistemic logic has been used to model classical information flows in distributed computation and multi-agent systems. They extended it to model quantum information flows, but their logic is not based on Birkhoff and von Neumann's quantum logic.
  Baltag and Smets \cite{Balt2012} also introduced quantum transition systems (QTS) and dynamic quantum logic to provide interpretation for these systems.
  This approach can be viewed as incorporating the modality of dynamic logic into the relational semantics for quantum logic proposed by Dishkant \cite{Dish1972} and Goldblatt \cite{Gold1974}.
  Recently, Holliday \cite{Holl2022,Holl2024b} and Holliday and Mandelkern \cite{Holl2024a} constructed a general semantics motivated by natural language studies to consider modalities over orthologic, applying it to epistemic logic. The relationship between their results and our system is mentioned in Section \ref{sec:qml}.

  As is well known, we can embed quantum logic into the normal modal logic {\bf B} using the relational semantics.
  Dalla Chiara \cite{Chia1977} considered the physical meaning of the modalities as follows: Each state $i$ is associated with a measure function $\mu_i$, which assigns probability values to propositions $\alpha$ of quantum logic.
  Let $\rho(\alpha)$ denote the embedding of $\alpha$ into $\mathbf{B}$, then the following holds:
  \begin{itemize}
    \item $i \models \Diamond\rho(\alpha)$ iff $\mu_i(\alpha) \ne 0$
    (the possibility of the modal embedding of $\alpha$ into $\mathbf{B}$ is true in the physical state $i$ iff $\alpha$ has a probability different from $0$ in $i$).
    \item $i \models \Box\rho(\alpha)$ iff $\mu_j(\alpha) = 1$ for any $j$.
    ($\alpha$ is a proposition that holds with certainty).
    \item $i \models \Box\Diamond\rho(\alpha)$ iff $\mu_i(\alpha) = 1$
    (the necessity of the possibility of the modal embedding of $\alpha$ into $\mathbf{B}$
    is true in the physical state $i$ iff  $\alpha$ has probability $1$ in $i$).
    \item $i \models \Box\Box\rho(\alpha)$ iff $i \models \Box\rho(\alpha)$
    \item $i \models \Diamond\Box\rho(\alpha)$ iff $i \models \Box\rho(\alpha)$
  \end{itemize}
  Note that this type of embedding represents quantum logic within a classical modal logic framework. Conversely, our quantum modal logic aims to encompass diverse modalities such as time, epistemicity, dynamicity, and beyond, within the realm of quantum logic.
  \subsection{Our Method}
  \label{sub:met}
  First, we combine the relational semantics for quantum logic and the Kripke semantics for modal logic.
  The relational semantics for quantum logic is due to Dishkant \cite{Dish1972}, Goldblatt \cite{Gold1974}, and Dalla Chiara \cite{Chia1977}. The idea of combining two logics is based on Gabbay \cite{Gabb1998}.
  In the process of this combination, the non-orthogonality relation of the relational semantics for quantum logic inevitably imposes a certain constraint on the accessibility relation of the Kripke semantics for modal logic.  Next, we axiomatize this semantics. Our axiomatic system is a minimal extension of Nishimura's sequent calculus for non-modal quantum logic \cite{Nish1980}. Finally, we rigorously prove the soundness and the completeness between the semantics and the axiomatic system.
  \section{Quantum modal logic}
  \label{sec:qml}
  This section introduces the formal definition of our quantum modal logic, denoted as $\mathbf{QML}$.
  \subsection{Language}
  \label{sub:lan}
  Our formal language employs the following symbols:
  \begin{itemize}
    \item Atomic formulas: $p$, $q$, $\dots$.
    \item Logical connectives: $\wedge$ (conjunction), $\neg$ (negation)
    \item Modal operators: $\Box$ (necessity)
  \end{itemize}
  Formulas are constructed in the usual way. We use lowercase Greek letters, such as $\alpha$, $\beta$, $\dots$, as meta-symbols for formulas. Let $F$ denote the set of all formulas. We use uppercase Greek letters, including those with subscripts, such as $\Gamma$, $\Delta$, $\dots$, and  $\Gamma_1$, $\Gamma_2$, $\dots$, as symbols for subsets of $F$.
  \begin{remark}
    As usual, we define $\vee$ (disjunction) as $\alpha \vee \beta \equiv \neg(\neg\alpha \wedge  \neg\beta)$ and $\Diamond$ (possibility) as $\Diamond\alpha \equiv \neg\Box\neg\alpha$.
    \qed
  \end{remark}
  In this paper, the following meta-logical symbols may be employed in proofs to improve readability:
  \begin{itemize}
    \item $\Rightarrow$: implies
    \item $\forall$: for all
    \item $\exists$: there exists
  \end{itemize}
  Remember that these symbols are not used as logical connectives of the object logic, i.e., $\mathbf{QML}$.
  \subsection{Semantics}
  \label{sub:str}
  The semantic definition of $\mathbf{QML}$ is as follows.
  \begin{definition}[Quantum modal structure]
    A {\it quantum modal structure} is defined as the following quadruple $\mathcal{S} = \langle W, R_Q, R_M, \rho \rangle$.
    \begin{itemize}
      \item $W$: a non-empty set.
      \item $R_Q$: a reflexive and symmetric relation on $W$.
      \item $R_M$: a binary relation on $W$, forced by $R_Q$.
      \item $\rho$:  an assignment of an $R_Q$-closed subset of $W$ to each atomic formula.
    \end{itemize}
    Let us explain each element of the quadruple.
    \begin{itemize}
      \item An element of $W$ is called a {\it world} or a {\it state}.
      \item $R_Q$ is called a {\it non-orthogonality relation} on $W$. This name derives from the non-or-thogonality relation based on the inner product of a Hilbert space.
      \item $R_M$ is called an {\it accessibility relation} with respect to the modal operators $\Box$ and $\Diamond$. This relation is forced by $R_Q$, meaning the following condition must be satisfied:
      \[R_M(i,l) \Rightarrow \forall j \in W (R_Q(i,j) \Rightarrow R_M(j,l)).\]
      Some background context for this condition is provided here. A reflexive and symmetric relation such as $R_Q$ is generally referred to as a similarity relation. Our motivation for using relational semantics to consider logic of physics stems from the desire to examine sets of similar physical situations and to single out the invariants associated with them \cite{Chia2002}. In light of this, our forcing condition can be read as follows: ``the set of situations that are seen through $R_M$ is invariant among similar situations.'' This  constraint is acceptable depending on the system under consideration.
      \item In the definition of $\rho$, a subset $X$ of $W$ is said to be {\it $R_Q$-closed} if it satisfies the following condition:
      \[i \in X \textrm{ iff } \forall j \in W(R_Q(i,j)  \Rightarrow \exists k \in W \textrm{ s.t. }(R_Q(j,k) \textrm{ and } k \in X)).\]
      The reason for considering the value of  $\rho$ as an $R_Q$-closed subset rather than an arbitrary subset of $W$ stems from the original idea in quantum logic that the set of all worlds making an experimental proposition true is not merely a set but a closed subspace of a Hilbert space.
    \end{itemize}
  \end{definition}
  \begin{remark}
    In Holliday \cite{Holl2022,Holl2024b} and Holliday and Mandelkern \cite{Holl2024a}, the forcing condition proposed there is more general than ours. Holliday's condition, expressed using our symbols for comparison, is as follows:
    if $R_M(i,l)$ and $R_Q(l,m)$, then $\exists i' \in W$ s.t. $R_Q(i,i')$ and $\forall j \in W(R_Q (i',j) \Rightarrow \exists l' \in W$ s.t. $R_M(j,l')$  and  $R_Q (l',m))$.
    That is, our forcing condition is a special case where  $i=i'$, $l=l'=m$ in Holliday's condition. This specialization contributes to obtaining a simplified axiomatic system.
  \end{remark}
  \begin{definition}[Truth]
    \label{df:tru}
    Let $\mathcal{S} = \langle W, R_Q, R_M, \rho \rangle$ be a quantum modal structure, $i \in W$,  and $\alpha \in F$. We define $\alpha$ to be true at $i$ in $\mathcal{S}$, denoted by $i \models_\mathcal{S} \alpha$, as follows:
    \begin{enumerate}
      \renewcommand{\labelenumi}{\alph{enumi})}
      \item $i \models_\mathcal{S} p$ iff $i \in \rho(p)$ for atomic formulas $p$
      \item $i \models_\mathcal{S} \alpha \wedge \beta$ iff $i \models_\mathcal{S} \alpha$ and $i \models_\mathcal{S} \beta$
      \item $i \models_\mathcal{S} \neg\alpha$ iff $\forall j \in W (R_Q(i,j) \Rightarrow j \not\models_\mathcal{S} \alpha)$
      \item $i \models_\mathcal{S} \Box\alpha$ iff $\forall l \in W (R_M(i,l) \Rightarrow l  \models_\mathcal{S} \alpha)$
    \end{enumerate}
  \end{definition}
  The following theorem demonstrates that the set of all worlds in which $\alpha$ is true becomes an $R_Q$-closed subset of $W$, not only when $\alpha$  is an atomic formula but also when it is any formula.
  The set of all $R_Q$-closed subsets forms an ortholattice, which is known as the algebraic structure of quantum logic.
  \begin{theorem}\label{th:bas}
    Let $\mathcal{S} = \langle W, R_Q, R_M, \rho \rangle$ be a quantum modal structure,  and $\alpha \in F$. Then,  the truth set of $\alpha$ in $\mathcal{S}$, i.e., $\{i \in W \mid i \models_\mathcal{S} \alpha\}$, is $R_Q$-closed.
  \end{theorem}
  \begin{proof}
    Let $i \in W$. We need to show that
    \[ i \models_\mathcal{S} \alpha \textrm{ iff } \forall j \in W (R_Q(i,j) \Rightarrow \exists k \in W \textrm{ s.t. } (R_Q(j,k) \textrm{ and } k \models_\mathcal{S} \alpha)). \]
    The `only if' part follows immediately from the symmetry of $R_Q$. In the following, we show the `if' part by induction on the construction of $\alpha$.
    \begin{enumerate}
      \renewcommand{\labelenumi}{\alph{enumi})}
      \item Case: Atomic formulas.
      Immediate by the definition of $\rho$.
      \item Case: $\alpha \wedge \beta$.
      We show that $i \models_\mathcal{S} \alpha \wedge \beta$ if $\forall j \in W \ (R_Q(i,j) \Rightarrow \exists k \in W$ s.t. $(R_Q(j,k)$ and $k \models_\mathcal{S} \alpha \wedge \beta))$.
      Suppose $\forall j \in W \ (R_Q(i,j) \Rightarrow \exists k \in W$ s.t. $(R_Q(j,k)$ and $k \models_\mathcal{S} \alpha \wedge \beta))$. Then, we have $\forall j \in W \ (R_Q(i,j) \Rightarrow \exists k \in W$ s.t. $(R_Q(j,k)$ and $k \models_\mathcal{S} \alpha$  and  $k \models_\mathcal{S} \beta))$ by the definition of truth.
      Suppose $R_Q(i,j)$. Then, we have $\exists k \in W$ s.t. $(R_Q(j,k)$ and $k \models_\mathcal{S} \alpha$ and $k \models_\mathcal{S} \beta)$.
      Let us refer to $k$, whose existence is claimed here, as $k_j$.
      Then, we have $R_Q(j,k_j)$, $k_j \models_\mathcal{S} \alpha$, and $k_j \models_\mathcal{S} \beta$. Hence,
      $\exists k \in W$ s.t. $(R_Q(j,k)$ and $k \models_\mathcal{S} \alpha)$
      and
      $\exists k \in W$ s.t. $(R_Q(j,k)$ and $k \models_\mathcal{S} \beta)$.
      That is,
      $R_Q(i,j) \Rightarrow \exists k \in W$ s.t. $(R_Q(j,k)$ and $k \models_\mathcal{S} \alpha)$
      and
      $R_Q(i,j) \Rightarrow \exists k \in W$ s.t. $(R_Q(j,k)$ and $k \models_\mathcal{S} \beta)$.
      Since $j$ is arbitrary, we have $\forall j \in W \ (R_Q(i,j) \Rightarrow \exists k \in W$ s.t. $(R_Q(j,k)$ and $k \models_\mathcal{S} \alpha))$ and $\forall j \in W \ (R_Q(i,j) \Rightarrow \exists k \in W$ s.t. $(R_Q(j,k)$ and $k \models_\mathcal{S} \beta))$.
      Then, we have $i \models_\mathcal{S} \alpha$ and $i \models_\mathcal{S} \beta$ by the induction hypothesis.
      Finally, we have $i \models_\mathcal{S} \alpha \wedge \beta$ by the definition of truth.
      \item Case: $\neg \alpha$.
      We show that $i \models_\mathcal{S} \neg \alpha$ if $\forall j \in W (R_Q(i,j) \Rightarrow \exists k \in W$ s.t. $(R_Q(j,k)$ and $k \models_\mathcal{S} \neg\alpha))$.
      Suppose $\forall j \in W (R_Q(i,j) \Rightarrow \exists k \in W$ s.t. $(R_Q(j,k)$ and $k \models_\mathcal{S} \neg\alpha))$.
      Then, we have $\forall j \in W (R_Q(i,j) \Rightarrow \exists k \in W$ s.t. $(R_Q(j,k)$ and $\forall l \in W(R_Q(k,l) \Rightarrow l \not\models_\mathcal{S} \alpha)))$ by the definition of truth.
      Then, we have
      $\forall j \in W (R_Q(i,j) \Rightarrow j \not\models_\mathcal{S} \alpha)$ by the induction hypotheis.
      Finally, we have $i \models_\mathcal{S} \neg \alpha$ by the definition of truth.
      \item Case: $\Box\alpha$.
      We show that $i \models_\mathcal{S} \Box\alpha$ if $\forall j \in W (R_Q(i,j) \Rightarrow \exists k \in W$ s.t. $(R_Q(j,k)$ and $k \models_\mathcal{S} \Box\alpha))$.
      Suppose $\forall j \in W (R_Q(i,j) \Rightarrow \exists k \in W$ s.t. $(R_Q(j,k)$ and $k \models_\mathcal{S} \Box\alpha))$.
      Then, we have $\forall j \in W (R_Q(i,j) \Rightarrow \exists k \in W$ s.t. $(R_Q(j,k)$ and
      $\forall l \in W (R_M(k,l) \Rightarrow l \models_\mathcal{S} \alpha)))$ (1) by the definition of truth.
      Suppose $R_M(i,l)$.
      Then, we have
      $\forall j \in W (R_Q(i,j) \Rightarrow R_M(j,l))$ (2)
      by the forcing condition.
      Let $j$ be $i$ in  (1). Since $R_Q(i,i)$ by the reflexivity of $R_Q$, we have
      $\exists k \in W$ s.t. $(R_Q(i,k)$ and
      $\forall l \in W (R_M(k,l) \Rightarrow l \models_\mathcal{S} \alpha))$.
      Let us refer to $k$, whose existence is claimed here, as $k_i$.
      Then, we have $R_Q(i,k_i)$ and
      $\forall l \in W (R_M(k_i,l) \Rightarrow l \models_\mathcal{S} \alpha)$ (3).
      Let $j$ be $k_i$ in (2). Since $R_Q(i,k_i)$, we have $R_M(k_i,l)$.
      Hence, $l \models_\mathcal{S} \alpha$ by (3).
      That is,
      $\forall l \in W (R_M(i,l) \Rightarrow l \models_\mathcal{S} \alpha)$.
      Finally, we have $i \models_\mathcal{S} \Box\alpha$
      by the definition of truth.
    \end{enumerate}
  \end{proof}
  Here is an example of the interpretation of a formula based on the above definition.
  \begin{example}
    To see if $i \models_\mathcal{S} \neg\Box p$, we need to check if  $\forall j \in W (R_Q(i,j) \Rightarrow j \not\models_\mathcal{S} \Box p)$. For this, assuming $R_Q(i.j)$, we need to check if $\exists k \in W$ s.t. $R_M(j,k)$ and $k \not\in\rho(p)$.
    \qed
  \end{example}
  The following remark shows that the meaning of $\Diamond$ is different from that in classical modal logic.
  \begin{remark}
    \label{re:dia}
    Since $i \models_\mathcal{S} \Diamond \alpha$ is the abbreviation of $i \models_\mathcal{S} \neg\Box\neg\alpha$, it is interpreted as $\forall j \in W (R_Q(i,j) \Rightarrow \exists k \in W$ s.t. ($R_M(j,k)$ and $\exists l \in W$ s.t. $(R_Q(k,l)$ and $l \models_\mathcal{S} \alpha )))$. This interpretation generally differs from the classical modal logic interpretation of $i \models_\mathcal{S} \Diamond \alpha$, which is  $\exists j \in W$ s.t. $(R_M(i,j)$ and $j \models_\mathcal{S} \alpha)$.
    Note that, according to Theorem \ref{th:bas}, this is equivalent to $\forall j \in W (R_Q(i,j) \Rightarrow \exists k \in W$ s.t. ($R_Q(j,k)$ and $\exists l \in W$ s.t. $(R_M(k,l)$ and $l \models_\mathcal{S} \alpha )))$.
    It is evident that both interpretations coincide when $R_Q = R_M$.
    \qed
  \end{remark}
  Let $\mathcal{S} = \langle W, R_Q, R_M, \rho \rangle$ be a quantum modal structure.
  We write  $\Gamma \models_\mathcal{S} \alpha$ to mean $\forall i \in W (i \models_\mathcal{S} \Gamma \Rightarrow i \models_\mathcal{S} \alpha)$. Here, $i \models_\mathcal{S} \Gamma$ denotes that $\forall \gamma \in \Gamma (i \models_\mathcal{S} \gamma)$. We write $\Gamma \models \alpha$ to mean $\Gamma \models_\mathcal{S} \alpha$ for any quantum modal structure $\mathcal{S}$. Finally, we write $\Gamma \models \Delta$ to mean $\exists \alpha \in \Delta$ s.t. $\Gamma \models \alpha$.
  \subsection{Axiomatization}
  \label{sub:axi}
  An axiomatic system for $\mathbf{QML}$ is presented below. The part concerning non-modal formulas is based on Nishimura \cite{Nish1980}.
  Expressions of the form $\Gamma \vdash \Delta$ in axioms and rules are called {\it sequents}, which represent the claim that at least one formula in $\Delta$ can be derived from a finite number of formulas in $\Gamma$.
  In general, we write (possibly empty) sets of formulas on the left and right sides of $\vdash$.
  We denote $\alpha, \Gamma$ and $\Gamma, \Delta$ as abbreviations for ${\{\alpha\}} \cup \Gamma$ and $\Gamma \cup \Delta$, respectively. Additionally, we denote $\neg \Gamma$ as $\{\neg \gamma \mid \gamma \in \Gamma\}$ and $\Box \Gamma$ as $\{\Box \gamma \mid \gamma \in \Gamma\}$.
  \begin{definition}[Axioms and rules]
    \begin{itemize}
      \item Axioms
      \begin{prooftree}
        \AxiomC{$\alpha \vdash \alpha$  \hspace{1.0mm} \scriptsize{(AX)}}
      \end{prooftree}
      \begin{prooftree}
        \AxiomC{$\Gamma \vdash \Box\alpha, \neg\Box\alpha$  \hspace{1.0mm}   \scriptsize{(MEM)}}
      \end{prooftree}
      \item Rules
      \begin{prooftree}
        \AxiomC{$\Gamma \vdash \Delta$}
        \RightLabel{\scriptsize{(WKN)}}
        \UnaryInfC{$\Pi, \Gamma \vdash \Delta, \Sigma$}
      \end{prooftree}
      \begin{prooftree}
        \AxiomC{$\Gamma_1 \vdash \Delta_1, \alpha$}
        \AxiomC{$\alpha, \Gamma_2 \vdash \Delta_2$}
        \RightLabel{\scriptsize{(CUT)}}
        \BinaryInfC{$\Gamma_1, \Gamma_2 \vdash  \Delta_1, \Delta_2$}
      \end{prooftree}
      \begin{prooftree}
        \AxiomC{$\alpha, \Gamma \vdash \Delta$}
        \RightLabel{\scriptsize{(${\rm \wedge {\rm l}_1}$)}}
        \UnaryInfC{$\alpha \wedge \beta, \Gamma \vdash \Delta$}
      \end{prooftree}
      \begin{prooftree}
        \AxiomC{$\beta, \Gamma \vdash \Delta$}
        \RightLabel{\scriptsize{(${\rm \wedge {\rm l}_2}$)}}
        \UnaryInfC{$\alpha \wedge \beta, \Gamma \vdash\Delta$}
      \end{prooftree}
      \begin{prooftree}
        \AxiomC{$\Gamma \vdash \Delta, \alpha$}
        \AxiomC{$\Gamma \vdash \Delta, \beta$}
        \RightLabel{\scriptsize{(${\rm \wedge r}$)}}
        \BinaryInfC{$\Gamma \vdash \Delta, \alpha \wedge \beta$}
      \end{prooftree}
      \begin{prooftree}
        \AxiomC{$\Gamma \vdash \Delta, \alpha$}
        \RightLabel{\scriptsize{(${\rm \neg l}$)}}
        \UnaryInfC{$\neg \alpha, \Gamma \vdash \Delta$}
      \end{prooftree}
      \begin{prooftree}
        \AxiomC{$\alpha \vdash \Delta$}
        \RightLabel{\scriptsize{(${\rm \neg r}$)}}
        \UnaryInfC{$\neg \Delta \vdash \neg \alpha$}
      \end{prooftree}
      \begin{prooftree}
        \AxiomC{$\alpha, \Gamma \vdash \Delta$}
        \RightLabel{\scriptsize{(${\rm \neg \neg l}$)}}
        \UnaryInfC{$\neg \neg \alpha, \Gamma \vdash \Delta$}
      \end{prooftree}
      \begin{prooftree}
        \AxiomC{$\Gamma \vdash \Delta, \alpha$}
        \RightLabel{\scriptsize{(${\rm \neg \neg r}$)}}
        \UnaryInfC{$\Gamma \vdash \Delta, \neg \neg \alpha$}
      \end{prooftree}
      \begin{prooftree}
        \AxiomC{$\Gamma \vdash \alpha$}
        \RightLabel{\scriptsize{({\bf K})}}
        \UnaryInfC{$\Box\Gamma \vdash \Box\alpha$}
      \end{prooftree}
    \end{itemize}
  \end{definition}
  \begin{remark}
    The axiomatic system for {\bf QML} consists of Nishimura's sequent calculus for non-modal quantum logic \cite{Nish1980}, along with the axiom MEM and the rule {\bf K}.
    MEM is a meta-logical version of the law of excluded middle that applies to formulas beginning with $\Box$, ensuring that either $\Box\alpha$ or $\neg\Box\alpha$ is derivable.
    {\bf K} is the usual rule also adopted in classical normal modal logics.
  \end{remark}
  \begin{remark}
    It can be shown as follows that our $\mathbf{QML}$ is an extension of non-modal orthologic.
    \begin{enumerate}
      \item Soundness and completeness are established between the semantics and the axiomatic system for $\mathbf{QML}$.
      \item Restricted to non-modal formulas, the axiomatic system for $\mathbf{QML}$ is equivalent to the system for Goldblatt's orthologic.
    \end{enumerate}
    We will demonstrate 1 in the next section, while 2 can be referenced in Nishimura \cite{Nish1980}.\qed
  \end{remark}
  In the following, we provide various definitions regarding the axiomatic system.
  A {\it derivation} is a finite sequence of sequents, where each sequent in the sequence is either an axiom or the lower sequent of a rule, with all upper sequents having already appeared in the sequence.
  We say that $\Gamma \vdash \Delta$ is {\it derivable}, or a {\it theorem}, if there exists a derivation where this sequent appears as the last element.
  Henceforth, we will simply use $\Gamma \vdash \Delta$ to mean that $\Gamma \vdash \Delta$ is derivable.
  We say that $\Gamma$ is {\it inconsistent} if there exists a formula $\alpha$ such that $\Gamma \vdash \alpha$ and $\Gamma \vdash \neg\alpha$.
  If $\Gamma$ is inconsistent, then $\Gamma \vdash \delta$ holds for any formula $\delta$.
  Indeed, The following derivation exists:
  \begin{prooftree}
    \AxiomC{$\Gamma \vdash \neg\alpha$}
    \AxiomC{$\Gamma \vdash \alpha$}
    \RightLabel{\scriptsize{($\neg{\rm l}$)}}
    \UnaryInfC{$\neg\alpha, \Gamma\vdash$}
    \RightLabel{\scriptsize{(${\rm CUT}$)}}
    \BinaryInfC{$\Gamma \vdash$}
    \RightLabel{\scriptsize{(WKN)}}
    \UnaryInfC{$\Gamma\vdash \delta$}
  \end{prooftree}
  We say that $\Gamma$ is {\it consistent} if it is not inconsistent.
  $\Gamma$ is consistent iff there exists some formula $\delta$ such that $\Gamma \not\vdash \delta$.
  A {\it deductive closure} of $\Gamma$, denoted $\overline{\Gamma}$, is defined as $\{ \gamma \mid \Gamma \vdash \gamma \}$.
  We say that $\Gamma$ is {\it deductively closed} if $\Gamma = \overline{\Gamma}$.
  The deductive closure of a consistent set remains consistent: If it were not, the original set would be inconsistent.
  We say that $\Gamma$ and $\Delta$ are {\it compatible} if $\forall \alpha \in F (\Gamma \vdash \alpha \Rightarrow \Delta \not\vdash \neg\alpha)$.
  Note that compatibility is symmetric: Suppose $\forall \alpha \in F (\Gamma \vdash \alpha \Rightarrow \Delta \not\vdash \neg\alpha)$. Then, we have $\forall \alpha \in F (\Delta\vdash \neg\neg\alpha \Rightarrow \Gamma \not\vdash \neg\alpha)$ (1).
  Now, suppose $\Delta \vdash \alpha$. Then, we have $\Delta \vdash \neg\neg\alpha$ by $\neg\neg$r.
  Finally, we have $\Gamma \not\vdash \neg\alpha$ by (1).

  The following theorem will play an important role in the proof of the completeness theorem in the next section.
  \begin{theorem}[Weak Lindenbaum theorem \cite{Chia2002}]\label{th:WLT}
    If $\Gamma \not\vdash \neg \alpha$, then there exists a consistent and deductively closed set $\Gamma^*$ of formulas that is compatible with $\Gamma$ and such that $\Gamma^* \vdash \alpha$.
  \end{theorem}
  \begin{proof}
    Suppose $\Gamma \not\vdash \neg \alpha$.
    Let $\Gamma^* \equiv \overline{\{\alpha\}}$. Then, we have $\Gamma^* \vdash \alpha$ by AX and WKN.
    Moreover, $\Gamma$ and $\Gamma^*$ are compatible:
    If they were not, there would exist some $\beta$ such that $\Gamma^*\vdash\beta$ and $\Gamma \vdash \neg\beta$, implying $\alpha \vdash \beta$ from $\Gamma^*\vdash\beta$, and $\neg\beta \vdash \neg\alpha$ by $\neg r$. This, along with $\Gamma \vdash \neg\beta$, would lead to $\Gamma \vdash \neg\alpha$ by CUT, resulting in a contradiction.
  \end{proof}
  \section{Soundness and Completeness}
  \label{sec:sac}
  In this section, we will provide proofs of two theorems (Theorem \ref{th:sou}, \ref{th:com}) that bridge the semantics and the axiomatic system for $\mathbf{QML}$.
  \subsection{Soundness}
  \label{sub:sou}
  \begin{theorem}[Soundness]
    \label{th:sou}
    $\Gamma \vdash \Delta \Rightarrow \Gamma \models \Delta$.
  \end{theorem}
  \begin{proof}
    Suppose $\Gamma \vdash \Delta$. Then, we have $\Gamma \vdash \alpha$ for some $\alpha \in \Delta$.
    Let $\mathcal{S}$ be a quantum modal structure.
    We show $\Gamma \models_\mathcal{S} \alpha$ by induction on the construction of the derivation.
    We proceed by cases according to the rule used at the end of the derivation of $\Gamma \vdash \alpha$.
    \begin{enumerate}
      \renewcommand{\labelenumi}{\alph{enumi})}
      \item Case: AX.
      It is evident that  $\forall i \in W (i \models_\mathcal{S} \alpha \Rightarrow i \models_\mathcal{S} \alpha)$. Hence, we have $\alpha \models_\mathcal{S} \alpha$.
      \item Case: MEM.
      Suppose $i \models_\mathcal{S} \Gamma$.
      Furthermore, suppose $i \not\models_\mathcal{S} \Box\alpha$, which means $\exists l \in W$ s.t. $(R_M(i,l)$ and $l \not\models_\mathcal{S} \alpha)$.
      Let us refer to $l$, whose existence is claimed here, as $l_i$.
      Now, suppose $R_Q(i,j)$, then we have $R_M(j,l_i)$ by the forcing condition on $R_M$.
      Therefore, $\forall j \in W(R_Q(i,j) \Rightarrow \exists l \in W$ s.t $(R_M(j,l)$ and $l \not \models_\mathcal{S} \alpha))$.
      Finally, we have  $i \models_\mathcal{S} \neg\Box\alpha$ by the definition of truth.
      \item Case: CUT.
      The case is divided into the following subcases i and ii.
      \begin{enumerate}
        \renewcommand{\labelenumii}{\roman{enumii})}
        \item
        Suppose that $\forall i \in W(i \models_\mathcal{S} \Gamma_1 \Rightarrow i \models_\mathcal{S} \delta_1)$ for some $\delta_1 \in \Delta_1$ (1).
        Furthermore, suppose $i \models_\mathcal{S} \Gamma_1 \cup \Gamma_2$. Then, we have $i \models_\mathcal{S} \delta_1$ by (1).
        \item Suppose that $\forall i \in W (i \models_\mathcal{S} \Gamma_1 \Rightarrow i \models_\mathcal{S} \alpha)$ (1) and that $\forall i \in W (i \models_\mathcal{S} \{\alpha\} \cup \Gamma_2 \Rightarrow i \models_\mathcal{S} \delta_2)$ for some $\delta_2 \in \Delta_2$ (2).
        Furthermore, suppose $i \models_\mathcal{S} \Gamma_1 \cup \Gamma_2$ (3).
        Then, we have $i \models_\mathcal{S} \alpha$ by (1).
        This, along with (3), leads to $i \models_\mathcal{S} \delta_2$ by (2).
      \end{enumerate}
      \item Case: $\wedge {\rm l}_1$.
      Let $\delta \in \Delta$.
      Suppose that $\forall i \in W (i \models_\mathcal{S} \{\alpha\} \cup \Gamma \Rightarrow i \models_\mathcal{S} \delta)$ (1).
      Furthermore, suppose $i \models_\mathcal{S} \{\alpha \wedge \beta\} \cup \Gamma$.
      Here, $i \models_\mathcal{S} \alpha \wedge \beta$ implies $i \models_\mathcal{S} \alpha$ by the definition of truth.
      Then, we have $i \models_\mathcal{S} \delta$ by (1).
      \item Case: $\wedge {\rm l}_2$.
      Same as above.
      \item Case: $\wedge$r.
      The case is divided into the following subcases i and ii.
      \begin{enumerate}
        \renewcommand{\labelenumii}{\roman{enumii})}
        \item Suppose that $\forall i \in W (i \models_\mathcal{S} \Gamma \Rightarrow i \models_\mathcal{S} \delta)$ for some $\delta \in \Delta$ (1).
        Furthermore, suppose $i \models_\mathcal{S} \Gamma$.
        Then, we have $i \models_\mathcal{S} \delta$ by (1).
        \item
        Suppose that $\forall i \in W (i \models_\mathcal{S} \Gamma \Rightarrow i \models_\mathcal{S} \alpha)$ (1) and that $\forall i \in W (i \models_\mathcal{S} \Gamma \Rightarrow i \models_\mathcal{S} \beta)$ (2).
        Furthermore, suppose $i \models_\mathcal{S} \Gamma$.
        Then, we have $i \models_\mathcal{S} \alpha$ by (1), and $i \models_\mathcal{S} \beta$ by (2). Finally, we have $i \models_\mathcal{S} \alpha \wedge \beta$ by the definition of truth.
      \end{enumerate}
      \item Case: $\neg$ l.
      The case is divided into the following subcases i and ii.
      \begin{enumerate}
        \renewcommand{\labelenumii}{\roman{enumii})}
        \item Suppose that $\forall i \in W (i \models_\mathcal{S} \Gamma \Rightarrow i \models_\mathcal{S} \delta)$ for some $\delta \in \Delta$ (1).
        Furthermore, suppose $i \models_\mathcal{S} \{\neg \alpha\} \cup \Gamma$.
        Then, we have $i \models_\mathcal{S} \delta$ by (1).
        \item Suppose that $\forall i \in W (i \models_\mathcal{S} \Gamma \Rightarrow i \models_\mathcal{S} \alpha)$ (1).
        Furthermore, suppose $i \models_\mathcal{S} \{\neg \alpha\} \cup \Gamma$ (2).
        Then, we have $i \models_\mathcal{S} \alpha$ by (1), and $\forall j \in W (R_Q(i,j) \Rightarrow j \not\models_\mathcal{S} \alpha)$ by (2). From the latter, we have $i \not\models_\mathcal{S} \alpha$ by the reflexivity of $R_Q$.
        Thus, we have both $i \models_\mathcal{S} \alpha$ and $i \not\models_\mathcal{S} \alpha$, which is a case that never occurs.
        Therefore, we have that $i \models_\mathcal{S} \{\neg \alpha\} \cup \Gamma \Rightarrow i \models_\mathcal{S} \delta$ for some (in fact, any) $\delta \in \Delta$.
      \end{enumerate}
      \item Case: $\neg$r.
      Let $\delta \in \Delta$.
      Suppose that $\forall i \in W (i \models_\mathcal{S} \alpha \Rightarrow i \models_\mathcal{S} \delta)$ (1).
      Furthermore, suppose $i \models_\mathcal{S} \neg\Delta$.
      Then, we have
      $i \models_\mathcal{S} \neg\delta$. That is, $\forall j \in W(R_Q(i,j) \Rightarrow j \not\models_\mathcal{S} \delta)$  (2).
      Suppose $R_Q(i,j)$. Then, we have $j \not\models_\mathcal{S} \delta$ by (2), and $j \not\models_\mathcal{S} \alpha$ by (1).
      Therefore, $\forall j \in W (R_Q(i,j) \Rightarrow j \not\models_\mathcal{S} \alpha)$.
      Finally, we have $i \models_\mathcal{S} \neg\alpha$ by the definition of truth.
      \item Case: $\neg\neg$ l.
      Let $\delta \in \Delta$.
      Suppose that $\forall i \in W (i \models_\mathcal{S} \{\alpha\} \cup \Gamma \Rightarrow i \models_\mathcal{S} \delta)$ (1).
      Furthermore, suppose $i \models_\mathcal{S} \{\neg\neg\alpha\} \cup \Gamma$.
      Here, $i \models_\mathcal{S} \neg\neg\alpha$ means that
      $\forall j \in W (R_Q(i,j) \Rightarrow j \not\models_\mathcal{S} \neg\alpha)$
      by the definition of truth.
      That is, $\forall j \in W (R_Q(i,j) \Rightarrow \exists k \in W$ s.t. $(R_Q(j,k)$ and $k \models_\mathcal{S} \alpha))$.
      Then, we have  $i \models_\mathcal{S} \alpha$ by Theorem \ref{th:bas}.
      Finally, we have $i \models_\mathcal{S} \delta$ by (1).
      \item Case: $\neg\neg$ r.
      The case is divided into the following subcases i and ii.
      \begin{enumerate}
        \renewcommand{\labelenumii}{\roman{enumii})}
        \item Suppose that $\forall i \in W (i \models_\mathcal{S} \Gamma \Rightarrow i \models_\mathcal{S} \delta)$ for some $\delta \in \Delta$ (1).
        Furthermore, suppose
        $i \models_\mathcal{S} \Gamma$.
        Then, we have $i \models_\mathcal{S} \delta$ by (1).
        \item  Suppose that $\forall i \in W (i \models_\mathcal{S} \Gamma \Rightarrow i \models_\mathcal{S} \alpha)$ (1).
        Furthermore, suppose $i \models_\mathcal{S} \Gamma$.
        Then, we have $i \models_\mathcal{S} \alpha$ by (1).
        That is, $\forall j \in W (R_Q(i,j) \Rightarrow \exists k \in W$ s.t. $(R_Q(j,k)$  and $k \models_\mathcal{S} \alpha))$ by Theorem \ref{th:bas}.
        This means that
        $\forall j \in W (R_Q(i,j) \Rightarrow j \not\models_\mathcal{S} \neg\alpha)$.
        Finally, we have $i \models_\mathcal{S} \neg\neg\alpha$ by the definition of truth.
      \end{enumerate}
      \item Case: ${\bf K}$.
      Suppose that $\forall i \in W (i \models_\mathcal{S} \Gamma \Rightarrow i \models_\mathcal{S} \alpha)$ (1).
      Furthermore, suppose $i \models_\mathcal{S} \Box\Gamma$.
      Then, we have $\forall \gamma \in \Gamma (i \models_\mathcal{S} \Box\gamma)$.
      That is,
      $\forall \gamma \in \Gamma (\forall j \in W (R_M(i,j) \Rightarrow j \models_\mathcal{S} \gamma))$.
      Suppose $R_M(i,j)$.
      Then, we have
      $\forall \gamma \in \Gamma(j \models _\mathcal{S} \gamma)$.
      This leads to $j \models_\mathcal{S} \alpha$ by (1).
      Therefore, we have
      $\forall j \in W (R_M(i,j) \Rightarrow j \models_\mathcal{S} \alpha)$.
      Finally, we have
      $i \models_\mathcal{S} \Box\alpha$
      by the definition of truth.
    \end{enumerate}
  \end{proof}
  \subsection{Completeness}
  \label{sub:com}
  The flow of the proof is based on Dalla Chiara and Giuntini \cite{Chia2002}, which addresses non-modal quantum logic.
  We construct the canonical model and verify that it is a quantum modal structure.
  In this model, truth corresponds to derivability.
  Therefore, if we assume that a formula is not derivable, it is false in the canonical model. Thus, completeness holds.
  \begin{definition}[Canonical model]
    We construct the canonical model $\mathcal{S_C}=\langle W, R_Q, R_M, \rho\rangle$ as follows:
    \begin{itemize}
      \item $W$: The set of all deductively closed consistent sets of formulas.
      \item $R_Q(i,j)$ iff $i$ and $j$ are compatible, i.e., $\forall \alpha  \in F (i \vdash \alpha \Rightarrow j \not\vdash \neg\alpha)$.
      \item $R_M(i,l)$ iff  $\forall \alpha \in F (i \vdash \Box\alpha \Rightarrow l \vdash \alpha)$.
      \item $\rho(p) \equiv \{i \in W \mid p \in i \}$ for atomic formulas $p$.
    \end{itemize}
  \end{definition}
  The following lemma demonstrates that the canonical model defined in this manner is indeed a quantum modal structure.
  \begin{lemma}
    Let $\mathcal{S_C}=\langle W, R_Q, R_M, \rho\rangle$ be the canonical model.
    Then, the following holds:
    \begin{itemize}
      \item $R_Q$ is reflexive and symmetric.
      \item $R_M$ is forced by $R_Q$.
      \item $\rho(p)$ is $R_Q$-closed for atomic formulas $p$.
    \end{itemize}
  \end{lemma}
  \begin{proof}
    Reflexivity of $R_Q$: Since any element $i$ in $W$ is consistent, $\forall \alpha \in F (i \vdash \alpha \Rightarrow i \not\vdash \neg\alpha)$.
    This means that $i$ is compatible with $i$.
    Symmetry of $R_Q$: Suppose that $\forall \alpha \in F (i \vdash \alpha \Rightarrow j \not\vdash \neg\alpha)$. Then, we have $\forall \alpha \in F (j \vdash \neg\neg\alpha \Rightarrow i \not\vdash \neg\alpha)$ (1).
    Now, suppose $j \vdash \alpha$. Then, we have $j \vdash \neg\neg\alpha$ by $\neg\neg r$. Therefore, we have $i \not\vdash \neg\alpha$ by (1).

    Forcing on $R_M$:
    We show that $R_M(i,l) \Rightarrow \forall j \in W (R_Q(i,j) \Rightarrow R_M(j,l))$. Suppose $R_M(i,l)$. Then, we have
    $\forall \alpha \in F (i \vdash \Box\alpha \Rightarrow l \vdash \alpha)$,
    hence $\forall \alpha \in F (l \not\vdash \alpha \Rightarrow i \not\vdash \Box\alpha)$ (2).
    Now, suppose $l \not\vdash \alpha$.
    Then, we have $i \not\vdash \Box\alpha$ by (2),
    implying $i \vdash \neg\Box\alpha$ by MEM.
    Furthermore, suppose $R_Q(i,j)$.
    Then, we have $j \not\vdash \neg\neg\Box\alpha$,
    implying $j \not\vdash \Box\alpha$ by $\neg\neg$r.
    Therefore, we conclude that $\forall \alpha \in F (l \not\vdash \alpha \Rightarrow j \not\vdash \Box\alpha)$, which means $R_M(j,l)$.

    $R_Q$-closedness:
    We show only the non-trivial implication here: $\forall j \in W(R_Q(i,j)  \Rightarrow \exists k \in W$ s.t. $(R_Q(j,k)$ and $k \in \rho(p))) \Rightarrow i \in \rho(p)$.
    Suppose $i \not\in \rho(p)$. Then, we have $p \not\in i$.
    Since $i$ is deductively closed, we have $i \not\vdash p$, implying $i \not\vdash \neg\neg p$ by $\neg\neg l$ and CUT.
    This, along with the weak Lindenbaum theorem (Theorem \ref{th:WLT}), leads to $\exists j \in W$ s.t. $R_Q(i,j)$ and $j \vdash \neg p$.
    Furthermore, suppose $R_Q(j,k)$.
    Then, we have $k \not\vdash \neg\neg p$,
    implying $k \not\vdash p$ by $\neg\neg$r.
    That is, $p \not\in k$, which means $k \not\in \rho(p)$.
    Therefore, we conclude that
    $\exists j \in W$ s.t. $R_Q(i,j)$ and $\forall k \in W (R_Q(j,k) \Rightarrow k \not\in \rho(p))$.\qed
  \end{proof}
  The following lemma is essential to the proof of the completeness theorem.
  \begin{lemma}\label{lem:comp}
    Let $\mathcal{S_C}=\langle W,R_Q,R_M,\rho \rangle$ be the canonical model for {\bf QML},
    $i \in W$, and $\alpha \in F$.
    Then, $i \models_\mathcal{S_C} \alpha$ iff $\alpha \in i$.
  \end{lemma}
  \begin{proof}
    We show both the `if' and the `only if' parts by induction on the construction of $\alpha$.
    \begin{enumerate}
      \renewcommand{\labelenumi}{\alph{enumi})}
      \item Case: Atomic formulas.
      Immediate by the definition of $\rho$.
      \item Case: $\alpha \wedge \beta$.
      First, we show the `if' part:
      $\alpha \wedge \beta\in i \Rightarrow i \models_\mathcal{S_C} \alpha \wedge  \beta$.
      Suppose $\alpha \wedge \beta\in i$.
      Then, we have
      $i \vdash \alpha \wedge \beta$.
      This leads to $i \vdash \alpha$ by $\wedge {\rm l}_1$ and CUT.
      Similary, $i \vdash \beta$.
      Since $i$ is deductively closed, we have
      $\alpha \in i$ and $\beta \in i$.
      By the induction hypothesis, we have $i \models_\mathcal{S_C} \alpha$ and $i \models_\mathcal{S_C} \beta$.
      Finally, we have $i \models_\mathcal{S_C} \alpha \wedge  \beta$
      by the definition of truth.
      Next, we show the `only if' part:
      $\alpha \wedge \beta \not\in i \Rightarrow i \not\models_\mathcal{S_C} \alpha \wedge \beta$.
      Suppose $\alpha \wedge \beta \not\in i$.
      Since $i$ is deductively closed, we have $i \not\vdash \alpha \wedge \beta$.
      That is, either
      $i \not\vdash \alpha$ or $i \not\vdash \beta$ by $\wedge$r.
      Suppose $i \not\vdash \alpha$.
      Then, we have $\alpha \not\in i$.
      By the induction hypothesis, we have $i \not\models_\mathcal{S_C} \alpha$.
      Finally, we have
      $i \not\models_\mathcal{S_C} \alpha \wedge \beta$
      by the definition of truth.
      The same reasoning applies to the case $i \not\vdash \beta$.
      \item Case: $\neg \alpha$.
      First, we show the `if' part:
      $\neg\alpha \in i \Rightarrow i \models_\mathcal{S_C} \neg\alpha$.
      Suppose
      $\neg \alpha \in i$.
      Then, we have $i \vdash \neg \alpha$.
      Furthermore, suppose $R_Q(i,j)$.
      Then, we have $j \not\vdash \neg\neg\alpha$,
      which leads to $j \not\vdash \alpha$ by $\neg\neg$r, and thus $\alpha \not\in j$.
      By the induction hypothesis, we have $j \not\models_\mathcal{S_C} \alpha$.
      Therefore, $\forall j \in W (R_Q(i,j) \Rightarrow j \not\models_\mathcal{S_C} \alpha)$.
      Finally, we have $i \models_\mathcal{S_C} \neg\alpha$ by the definition of truth.
      Next, we show the `only if' part:
      $\neg\alpha \not\in i \Rightarrow i \not\models_\mathcal{S_C} \neg\alpha$.
      Suppose $\neg\alpha \not\in i$.
      Since $i$ is deductively closed, we have
      $i \not\vdash \neg\alpha$.
      Then, we have
      $\exists j \in W$ s.t. $(R_Q(i,j)$ and $j \vdash \alpha)$
      by the weak Lindenbaum theorem (Theorem \ref{th:WLT}).
      Let us refer to $j$, whose existence is claimed here, as $j_i$.
      Since $j_i$ is deductively closed, we have $\alpha \in j_i$.
      By the induction hypothesis, we have
      $j_i \models_\mathcal{S_C} \alpha$.
      Therefore,
      $\exists j \in W$ s.t. $(R_Q(i,j)$ and $j \models_\mathcal{S_C} \alpha)$.
      Finally, we have
      $i \not\models_\mathcal{S_C} \neg\alpha$ by the definition of truth.
      \item Case: $\Box \alpha$.
      First, we show the `if' part:
      $\Box\alpha \in i \Rightarrow i \models_\mathcal{S_C} \Box\alpha$.
      Suppose $\Box\alpha \in i$.
      Then, we have $i \vdash \Box\alpha$.
      Furthermore, suppose $R_M(i,l)$.
      Then, we have $l \vdash \alpha$.
      Since $l$ is deductively closed, we have
      $\alpha \in l$.
      By the induction hypothesis, we have
      $l \models_\mathcal{S_C} \alpha$.
      Therefore,
      $\forall l \in W (R_M(i,l) \Rightarrow l \models_\mathcal{S_C} \alpha)$.
      Finally, we have $i \models_\mathcal{S_C} \Box\alpha$ by the definition of truth.
      Next, we show the `only if' part:
      $\Box\alpha \not\in i \Rightarrow i \not\models_\mathcal{S_C} \Box\alpha$.
      Suppose
      $\Box\alpha \not\in i$.
      Since $i$ is deductively closed, we have $i \not\vdash \Box\alpha$.
      Let $l_i := \{\gamma \mid \Box\gamma \in i\}$
      and suppose
      $l_i \vdash \alpha$.
      Then, we have $\gamma_1, \dots, \gamma_n \vdash \alpha$ for some $\gamma_1, \dots, \gamma_n \in l_i$.
      Now, applying $\mathbf{K}$ to this, we have $\Box\gamma_1, \dots, \Box\gamma_n \vdash \Box\alpha$. Consequently, we obtain $i \vdash \Box\alpha$, leading to a contradiction.
      Hence, we have $l_i \not\vdash \alpha$.
      Then, we have $R_M(i,\overline{l}_i)$ and $\overline{l}_i \not\vdash \alpha$.
      The latter means
      $\alpha \not\in \overline{l}_i$,
      which leads to $\overline{l}_i \not\models_\mathcal{S_C} \alpha$
      by the induction hypothesis.
      Therefore,
      $\exists l \in W$ s.t. $(R_M(i,l)$ and $l \not\models_\mathcal{S_C} \alpha)$.
      Finally, we have $i \not\models_\mathcal{S_C} \Box\alpha$
      by the definition of truth.
    \end{enumerate}
  \end{proof}
  \begin{theorem}[Completeness]
    \label{th:com}
    $\Gamma \models \alpha \Rightarrow \Gamma \vdash \alpha$.
  \end{theorem}
  \begin{proof}
    Suppose $\Gamma \not\vdash \alpha$.
    This implies $\Gamma$ is consistent.
    Let $i := \overline{\Gamma}$.
    Then, we have $i \in W$ and $\forall \gamma \in \Gamma (\gamma \in i)$.
    Hence, $i \models_\mathcal{S_C} \Gamma$ by  Lemma \ref{lem:comp}.
    Now, suppose
    $i \models_\mathcal{S_C} \alpha$.
    Then, we have $\alpha \in i$  by
    Lemma \ref{lem:comp}.
    This leads to  $\Gamma \vdash \alpha$ by the definition of $i$,
    which is a contradiction.
    Hence, we have $i \not\models_\mathcal{S_C} \alpha$.
    Therefore,
    $\Gamma \not\models_\mathcal{S_C} \alpha$.
    \qed
  \end{proof}
  \section{Applications}
  \label{sec:app}
  Depending on how we define the relation $R_M$ and interpret the modal operators, we can obtain various systems of quantum modal logic.
  In this section, we will briefly explore several examples. Note that in these examples, the modal operators allowed by our semantics are limited to those that satisfy our forcing condition on $R_M$: the worlds connected by $R_Q$ can see the same set of worlds through $R_M$.
  \paragraph{Alethic logic}
  Dalla Chiara \cite{Chia1977}, as mentioned in Section \ref{sub:bac}, was considering a physical interpretation of alethic modalities, such as necessity and possibility.
  In {\bf QML}, alethic modality over quantum logic can be realized by setting $R_M = R_Q$.
  Note that in this case, according to Remark \ref{re:dia}, the usual semantics for $\Diamond$  holds, i.e., $i \models_\mathcal{S} \Diamond\alpha $ iff $\exists j \in W$ s.t. $(R_M(i,j)$ and  $j \models_\mathcal{S} \alpha)$.
  Alethic modalities can be physically interpreted as follows:
  \begin{itemize}
    \item $i \models_\mathcal{S} \Box\alpha$: $\alpha$ is a proposition that holds with certainty in any state.
    \item $i \models_\mathcal{S} \Diamond\alpha$: $\alpha$ has a probability different from $0$ in the state $i$.
    \item $i \models_\mathcal{S} \Box\Diamond\alpha$: $\alpha$ has probability $1$ in the  state $i$.
  \end{itemize}
  $R_Q$ is originally reflexive and symmetric, but it also acquires transitivity due to its own forcing condition. Consequently, $R_Q$ becomes an equivalence relation, and the following propositions stated in Section \ref{sub:bac} also hold.
  \begin{itemize}
    \item $i \models_\mathcal{S} \Box\Box\rho(\alpha)$ iff $i \models_\mathcal{S} \Box\rho(\alpha)$
    \item $i \models_\mathcal{S} \Diamond\Box\rho(\alpha)$ iff $i \models_\mathcal{S} \Box\rho(\alpha)$
  \end{itemize}
  \paragraph{Temporal logic}
  We can introduce a structure of time by defining $R_M$ as a strict total order, that is, a non-reflexive total order. Then, a modal formula $\Box\alpha$ means that $\alpha$ always holds in the future. At this point, the forcing condition imposed on $R_M$ implies that the set of possible future situations is invariant among similar situations.
  As in ordinary temporal logic, there are potential variations of quantum temporal logic, including discrete-time or continuous-time, linear-time or branching-time, and more. We may implement these variations by appropriately defining $W$ and $R_M$.
  \paragraph{Dynamic logic}
  To replicate the concept of dynamic quantum logic proposed by Baltag and Smets \cite{Balt2012}, we need to consider a multi-modal logic where each {\it action} \(P\) is associated with a modal operator \([P]\) and an accessibility relation \(R_M^{P}\).
  There is the following link between $R_Q$ and $R_M^P$:
  $R_Q(i,j)$ iff $ \exists{P}$ s.t.  $R_M^{P}(i,j)$.
  Actions in this framework include tests, i.e., measurements, and unitary evolution. Unfortunately, these actions do not generally satisfy our forcing condition. Therefore, our next goal should be to relax this constraint while still obtaining a concise axiomatic system.
  \section{Concluding remark}
  \label{sec:con}
  In this paper, we have presented the semantics and the axiomatic system for quantum modal logic, and have rigorously established the soundness and the completeness theorems. Although the obtained system is quite simple, it requires a relatively strong condition on $R_M$. Referring to the recent results \cite{Holl2022,Holl2024a,Holl2024b}, we will explore the possibility of a more applicable and tractable axiomatic framework in the forthcoming research.
  \section*{Acknowledgements}
  We appreciate the editors and the anonymous reviewers for their thorough examination and numerous insightful comments, which helped us correct errors and significantly improve this paper.
  \section*{Funding}
  This work was supported by Japan Society for the Promotion of Science (JSPS) KAKENHI [Grant Number JP24K03372].
  

\begin{thebibliography}{99}
    \bibitem{Balt2010}
    Baltag, A and S. Smets. `Correlated knowledge: an epistemic-logic view on quantum entanglement.' \emph{International Journal of Theoretical Physics} 49: 3005-3021, 2010.
    \bibitem{Balt2012}
    Baltag, A. and S. Smets. `The dynamic turn in quantum logic.' \emph{Synthese} 186: 753-773,  2012.
    \bibitem{Birk1936}
    Birkhoff,~G. and J.~von Neumann, `The Logic of Quantum Mechanics.', \emph{Annals of Mathematics} 37(4):823--843, 1936.
    \bibitem{Chia1977}
    Dalla Chiara, M. L. `Quantum logic and physical modalities.'
    \emph{Journal of Philosophical Logic} 6(1): 391-404, 1977.
    \bibitem{Chia2002}
    Dalla Chiara, M. L., and R. Giuntini. `Quantum logics.' \emph{Handbook of philosophical logic} 129-228, 2002.
    \bibitem{Dish1972}
    Dishkant, H. `Semantics of the minimal logic of quantum mechanics.', \emph{Studia Logica: An International Journal for Symbolic Logic} 30:23-32, 1972.
    \bibitem{Dish1977}
    Dishkant, H. `Imbedding of the quantum logic in the modal system of Brower.' \emph{The Journal of Symbolic Logic} 42(3): 321-328, 1977.
    \bibitem{Dish1978}
    Dishkant, H. `An extension of the \L{}ukasiewicz logic to the modal logic of quantum mechanics.' \emph{Studia Logica} 37: 149-155,  1978.
    \bibitem{Gabb1998}
    Gabbay, Dov M. `\emph{Fibring logics}.' Vol. 38. Clarendon Press, 1998.
    \bibitem{Gold1974}
    Goldblatt, R. I. `Semantic analysis of orthologic.' \emph{Journal of Philosophical logic} 3(1/2):19-35, 1974.
    \bibitem{Holl2022}
    Holliday, W. H. `Compatibility and accessibility: lattice representations for semantics of non-classical and modal logics.' \emph{arXiv preprint} arXiv:2201.07098, 2022.
    \bibitem{Holl2024a}
    Holliday, W. H., and M. Mandelkern. `The orthologic of epistemic modals.' \emph{Journal of Philosophical Logic} 53: 831-907, 2024.
    \bibitem{Holl2024b}
    Holliday, W. H. `Modal logic, fundamentally.' \emph{arXiv preprint} arXiv:2403.14043, 2024.
    \bibitem{Nish1980}
    Nishimura, H. `Sequential method in quantum logic.'
    \emph{The Journal of Symbolic Logic} 45.2: 339-352,  1980.
    \bibitem{Mack1963}
    Mackey, G. W. `The mathematical foundations of quantum mechanics,' eds Benjamin WA Inc., New York, 1963.
    \bibitem{Mitt1979}
    Mittelstaedt, P. `The modal logic of quantum logic.' \emph{Journal of Philosophical Logic} 8(1): 479-504, 1979.
    \bibitem{Pess2005}
    Pessoa Jr, O. `Towards a modal logical treatment of quantum physics.' \emph{Logic Journal of IGPL} 13(1): 139-147,  2005.
    \bibitem{Fraa1981}
    van Fraassen, B. C. `A modal interpretation of quantum mechanics.' \emph{Current issues in quantum logica}:229-258, Boston, MA: Springer US, 1981.
    \bibitem{Yu2019}
    Yu, N. `Quantum temporal logic.' arXiv preprint arXiv:1908.00158, 2019.
  \end{thebibliography}
\end{document}